\documentclass[conference]{IEEEtran}
\IEEEoverridecommandlockouts

\usepackage{xcolor}
\definecolor{DarkGreen}{rgb}{0.1,0.5,0.1}
\definecolor{DarkRed}{rgb}{0.5,0.1,0.1}
\definecolor{DarkBlue}{rgb}{0.1,0.1,0.5}

\usepackage[small]{caption}
\usepackage{amssymb,amsfonts, amsthm,graphicx}
\usepackage[cmex10]{amsmath}
\usepackage{nicefrac,comment}

\newcommand{\cC}{\ensuremath{\mathcal{C}}}

\newcommand{\F}{{\mathbb F}}

\newcommand{\inset}[1]{\left\{#1\right\}}

\newcommand{\inparen}[1]{\left(#1\right)}

\newcommand{\eps}{\varepsilon}
\renewcommand{\epsilon}{\varepsilon}

\newtheorem{theorem}{Theorem} 
\newtheorem{lemma}[theorem]{Lemma} 
\newtheorem{definition}{Definition}

\newtheorem{claim}[theorem]{Claim}

\newtheorem{question}{Question}

\title{On taking advantage of multiple requests \\ in error correcting codes\thanks{This work is partially supported by NSF Grant CCF-1657049. }}

\author{

\IEEEauthorblockN{Prasanna Ramakrishnan and Mary Wootters}
\IEEEauthorblockA{
\textit{Stanford University}\\
Stanford, CA \\
\{pras1712,marykw\}@stanford.edu}

}

\newcommand{\codename}{robust batch code }
\newcommand{\codenamens}{robust batch code} 

\newcommand{\Codenamens}{Robust batch code}

\begin{document}
\maketitle

\begin{abstract}
In most notions of locality in error correcting codes---notably locally recoverable codes (LRCs) and locally decodable codes (LDCs)---a decoder seeks to learn a single symbol of a message while looking at only a few symbols of the corresponding codeword.  However, suppose that one wants to recover $r > 1$ symbols of the message.  The two extremes are repeating the single-query algorithm $r$ times (this is the intuition behind LRCs with availability, primitive multiset batch codes, and PIR codes) or simply running a global decoding algorithm to recover the whole thing.  In this paper, we investigate what can happen in between these two extremes: at what value of $r$ does repetition stop being a good idea?

In order to begin to study this question we introduce \em \codenamens{s}, \em which seek to find $r$ symbols of the message using $m$ queries to the codeword, in the presence of erasures.  We focus on the case where $r=m$, which can be seen as a generalization of the MDS property.  Surprisingly, we show that for this notion of locality, repetition is optimal even up to very large values of $r = \Omega(k)$.
\end{abstract}

\section{Introduction}
In a traditional coding theory setup, a message $x \in \F^k$ is encoded as a codeword $c \in \F^n$, for $n \geq k$.  The goal is to be able to recover $x$ given a corrupted version of $c$.  An important notion in coding theory is \em locality, \em which arises in applications ranging from complexity theory to distributed storage.  
Locality refers to the ability to recover a single symbol $x_i$ of the message (we call this a \em request\em) without looking at too many codeword symbols (these are \em queries\em).
There are several ways of setting up the problem: in \em locally decodable codes \em (LDCs), one wants to recover $x_i$ despite a constant fraction of errors in $c$.  In the setting of \em locally recoverable codes \em (LRCs), one typically wants to recover $x_i$ despite a small number of erasures in $c$.  

In this paper, we focus on locality in the setting where we request not just one symbol, but $1 < r < k$ symbols.
There are notions---like \em batch codes, \em or \em LRCs with availability\em---that focus on this.  However, for these notions, the primary way of obtaining codes which can recover $r$ symbols is by designing codes which allow for the repeated parallel recovery of a single symbol.  For small values of $r$, this is a reasonable thing to do, and it seems unlikely that one can do better.  However, for larger values of $r$, one can clearly do better: as an extreme example, when $r = k$, one can simply (globally) decode the whole codeword and this would be better then locally decoding each coordinate separately.

The question which we ask in this work is when this shift occurs.  That is, 

\begin{question}\label{q:main}
How large does $r$ have to be before repeating the single-request protocol $r$ times is substantially sub-optimal?
\end{question}

Answering Question~\ref{q:main} seems very difficult for LDCs, where we do not even understand the single-request case.  Instead, we focus on the model where some number $d$ of codeword symbols have been adversarially erased.  To this end, we define a new variant of an LRC, which we call a \em \codenamens. \em 
\Codenamens{s} are error correcting codes which have the guarantee that any $r$ message symbols can be recovered by looking at $m$ symbols of the codeword, even in the presence of $d$ erasures.  The goal is the obtain the smallest block length $n$, given $r,m,d$ and the length of the message $k$. 

Even restricted to the erasure setting, this question still seems quite difficult.  In this work, we restrict our attention further to the setting where $r = m$.  That is, we would like to recover $r$ message symbols by reading exactly $r$ codeword symbols.  This setting provides a very nice testing ground to begin to explore Question~\ref{q:main}, because the extreme cases have clear answers:
\begin{itemize}
	\item When $r=m=1$, we are only allowed one query and clearly we cannot do any better than the repetition code which repeats the message $d+1$ times.   In this case, we have $n = k(d + 1)$.
	\item When $r=k$, then we can use an MDS code, for example a Reed-Solomon code, to encode $x$.  In this case, we achieve $n = k + d$, which is optimal by the Singleton bound.
\end{itemize}

Thus, Question~\ref{q:main} becomes: at what $r$ does $n = k(d+1)$ stop being optimal?  
Since there is a much better solution of $n = k + d$ when $r = k$, we might expect to be able to do at least slightly better than repetition for values of $r$ that are, say, $\sqrt{k}$ or even $k/10$.  However, our main result (formally stated in Theorem~\ref{thm:main}) is that this is not the case.  In natural parameter regimes, even when $r = \Omega(k)$, the repetition code is optimal!  

The natural question left by our work is what happens when we allow $m > r$?  For example, what if $m$ may be as large as $10r$?  It seems possible that the techniques we develop here may be helpful in proving impossibility results in this setting as well.  On the other hand, a construction of such codes with $m > r$ would be a very useful building block in the design of codes for distributed storage, private information retrieval, and other areas.

\textbf{Outline.} In Section~\ref{sec:defs} we lay out notation and preliminary definitions.  We survey related work in Section~\ref{sec:relatedwork}.  In Section~\ref{sec:main}, we formally state our main result (Theorem~\ref{thm:main}) and discuss its implications.  We prove Theorem~\ref{thm:main} in Section~\ref{sec:proof}. 

\section{\Codenamens{s}}\label{sec:defs}

First, we lay out some notation.
We use $[n]$ to denote the set $\{1,\ldots,n\}$, and $\F$ to denote an arbitrary finite field.  Our results hold over any finite field, but for our motivation we imagine that $\F$ is large enough that MDS codes exist.  For a vector $v \in \F^n$ and a set $I \subset[n]$, we use $v|_I \in \F^{|I|}$ to denote the restriction of $v$ to the indices $i \in I$.  Similarly, for a matrix $G \in \F^{k \times n}$, and subsets $I \subset[k], j \subseteq[n]$, we use $G|_{I,J}$ to denote the matrix $G$ restricted to the rows indexed by $i \in I$ and the columns indexed by $j \in J$.
Throughout, we will consider \em linear codes. \em  A linear code can be represented by its \em generator matrix \em $G \in \F^{k \times n}$: the encoding $\cC(x)$ of the message $x \in \F^k$ is given by $xG \in \F^n$.  The \em rate \em of a code $\cC: \F^k \to \F^n$ is defined to be the ratio $k/n$.

\begin{definition}\label{def:main}
A $(r,m,d)$-\codename with block length $n$ and message length $k$ over a field $\F$ is a linear map $\cC:\F^k \to \F^n$, so that the following holds.  For any $x \in \F^k$, and for any set $I \subseteq [k]$ of size $r$ and any set $D \subseteq [n]$ of size $d$, there is some set $J \subseteq [n]$ of size at most $m$ so that $D \cap J = \emptyset$, and so that $\cC(x)|_{J}$ determines $x|_I$.
\end{definition}

Notice that a $(r,m,d)$-\codename has distance at least $d$, because every coordinate $x_i$ of the message can be recovered from any $d$ erasures, hence the whole codeword can be recovered from any $d$ erasures.

In this paper we will focus the special case where $m = r$.  There are two reasons for doing this.  First, it it an interesting notion in and of itself: when $m = r = k$, Definition~\ref{def:main} is precisely the definition of an Maximum Distance Separable (MDS) code.  
That is, one way of defining an MDS code is to say that any $k$ symbols of the message should be obtainable from some $k$ symbols of the codeword, despite $d$ adversarial erasures, while $n$ is as small as possible given $k$ and $d$ (in this case, $n = k + d$).  In a $(r,r,d)$-\codenamens, we would like for any $r$ symbols of the message to be obtainable from some $r$ symbols of the codeword, despite $d$ adversarial erasures, while $n$ is as small as possible given $k$ and $d$.  In this sense, and $(r,r,d)$-\codename is perhaps a ``local" sort of MDS requirement.

The second reason we consider the case where $m = r$ is that it provides a clean setting in which to approach Question~\ref{q:main}; as discussed in the introduction, the extreme cases of $r=1$ and $r=k$ are well-understood and quite different, so it is natural to ask where the transition occurs.

\section{Related Work}\label{sec:relatedwork}
There are many notions of ``batched" queries in the coding theory literature.  In this section, we give a brief overview and explain why these notions don't answer Question~\ref{q:main}. 

\paragraph{LRCs with availability and related notions}
Locally recoverable codes (LRCs) and similar notions have been heavily studied over the past five years or so~\cite{GHSY12,HCL13,PD14,TB14}, motivated by applications in distributed storage.  In the LRC model, the goal is to design a code with good distance, and so that each message symbol $x_i$ (or more generally, each codeword symbol $c_i$) is a function of at most $m$ symbols in the codeword.
Thus, our work could be seen as tackling the problem of designing LRCs robust to $r$ failures.

Handling multiple failures in LRCs is also well-studied; the primary approach is via LRCs with \em multiple repair sets \em or \em availability\em~\cite{PHO13,WZ14,TB14,RPDV16}.  
In addition to handling multiple failures, LRCs with availability support multiple parallel requests for the same message symbol.  
More precisely, they have the additional requirement that each symbol $x_i$ has $t$ disjoint repair groups; that is, $t$ disjoint sets of codeword symbols that determine $x_i$.  It is easy to see that if $t=d+1$ then this implies \codenamens{s}:
indeed, even after $d$ erasures, each symbol has at least one of its $d+1$ repair sets intact, so we may recover all $r$ requests using at most $m = r \cdot m'$ codeword symbols, where $m'$ is the locality of the LRC.
However, this notion of $t$ disjoint repair groups is much stronger than the notion of a \codenamens, and the reduction above seems wasteful: it boils down to repeating a single request $r$ times.  Question~\ref{q:main} asks when we can do better.

There are many notions related to LRCs with availability.  In particular, PIR codes~\cite{codedpir}, Primitive Multiset Batch Codes~\cite{IKOS04,DGRS14}, and $t$-DRGP codes~\cite{FGW17} all work on basically the same principle of disjoint repair groups.
We refer the reader to~\cite{Ska16} for an excellent survey on LRCs, Batch Codes, and PIR codes.
As above, all of these codes do give rise to \codenamens{s}, but the reduction seems wasteful, essentially repeating the procedure for a single request $r$ times.  Rather than trying to construct many disjoint repair groups, a \codename should take advantage of overlapping repair groups between different symbols.

\paragraph{Batch codes}  
Batch codes, introduced in \cite{IKOS04}, are similar in spirit to our work: the original goal was to do Private Information Retrieval (PIR) for $r$ queries in a more efficient way than by repeating one query $r$ times.  In a batch code, the codeword symbols are divided into \em buckets, \em and the goal is to answer $r$ requests using at most $t$ queries from each bucket.  
However, to the best of our knowledge, the best constructions of batch codes go through \em primitive multiset batch codes, \em discussed above.  These essentially enforce the ``$t$-disjoint-repair-groups" requirement, and boil down to repeating the algorithm to handle a single request $r$ times.

\paragraph{Regenerating codes}
\Codenamens{s} can be seen a strengthening of regenerating codes.  Regenerating codes (see \cite{survey} for a survey) are codes which aim to minimize \em network bandwidth \em rather than locality.  In a systematic repair-by-transfer regenerating code, a \em node \em containing $r$ message symbols fails, and the goal is to recover it by looking at very few symbols in the rest of the code, appropriately spread out over surviving nodes.  Thus, such regenerating codes are \codenamens{s}, where the set $D$ of deletions is equal to the set $I$ that must be recovered, and is restricted to very special sets (the symbols contained within single nodes).

\paragraph{Locally decodable codes}
In the LDC setting (see \cite{Y11} for a survey) we seek to recover a single message symbol in the presence of errors, rather than erasures. Note that since the errors are adversarial and undetectable, the query algorithm of an LDC is necessarily randomized. 
There are constructions of high-rate LDCs~\cite{KSY10,GKS13,HOW13,KMRS16,HWZ17} with query complexity $n^\eps$ or even sub-polynomial in $n$.  
In fact, an LDC with query complexity $m$ is easily seen to have $t = \Omega(n/m)$ disjoint repair groups for every symbol (see, e.g., \cite{FGW17}), and thus high-rate LDCs also give rise to \codenamens{s}.
However, as above, the method for dealing with multiple failures here is just repeating a single request $r$ times.

In this work, we prove lower bounds (impossibility results) for \codenamens{s}, and it is natural to ask if our results have any implications for lower bounds in LDCs, a notoriously difficult problem.  Unfortunately, they do not.  The results we show in this paper are for the case where $r = m$.  In the LDC setting, this corresponds to a lower bound on $1$-query LDCs, which is not very interesting.  
Moreover, even with $m > r$, any interesting lower bounds for \codenamens{s} with $r=1$ would not translate into interesting lower bounds for LDCs because it seems unlikely that handling erasures is as hard as handling errors; and any results for $r > 1$ would be less relevant.

The work perhaps closest in spirit to ours (that we are aware of) is in the LDC setting: a recent work \cite{CXY16} on multi-point local decoding of Reed-Muller codes.
The punchline is that, in a particular parameter regime, Reed-Muller codes can recover $r$ requests (in the presence of errors) more effectively than by repeating the one-query procedure $r$ times.  This work differs from ours in that it focuses on Reed-Muller codes in particular (rather than general bounds), and additionally focuses on errors rather than erasures.

\section{Repetition is optimal most of the time}\label{sec:main}

In this section, we state our main result and discuss the implications in a few different parameter regimes. 

\begin{theorem}\label{thm:main}
Let $\cC:\F^k \to \F^n$ be a $(r,r,d)$-\codenamens.  Then
\[ n \geq k(d+1) - \max \inset{0, d(r-1) - \frac{1}{2}(k-r)^2 }. \]
\end{theorem}
Theorem~\ref{thm:main} implies that $n \geq k(d+1)$ (that is, the repetition code is optimal) whenever 
\begin{equation}\label{eq:condition}
r \leq k + d - \sqrt{(k + d)^2 - k^2}.
\end{equation}
Otherwise, we still have
\begin{equation}\label{eq:bound}
n \geq k(d + 1) - d(r - 1) + \frac{1}{2}(k - r)^2.
\end{equation}
Notice that when $r=k$, \eqref{eq:bound} reads $n \geq k + d$, which is precisely the Singleton bound.  Thus Theorem~\ref{thm:main} is tight when \eqref{eq:condition} holds (attained by a repetition code) and when $r=k$ (attained by Reed-Solomon codes).

Before we prove Theorem~\ref{thm:main}, we specialize it to a few different parameter regimes, to demonstrate that for very natural parameter regimes, for most values of $r$, one cannot do substantially better than repetition (or in some cases any better at all).
That is, when $r=m$, the answer to Question~\ref{q:main} is ``quite large."

To build intuition, Figure~\ref{fig:plt} shows the lower bound given in Theorem~\ref{thm:main} for $k=100$ for various values of $d$. Notice that it makes sense to take $d \ll k$ (this is the case in the LRC set-up) or $d \gg k$ (this is the case in the constant-query LDC set-up).  
The take-away from Figure~\ref{fig:plt} should be that the flat initial segment---where the repetition code is optimal---extends until $r$ is reasonably large.

\begin{figure}
\begin{center}
\includegraphics[width=8cm]{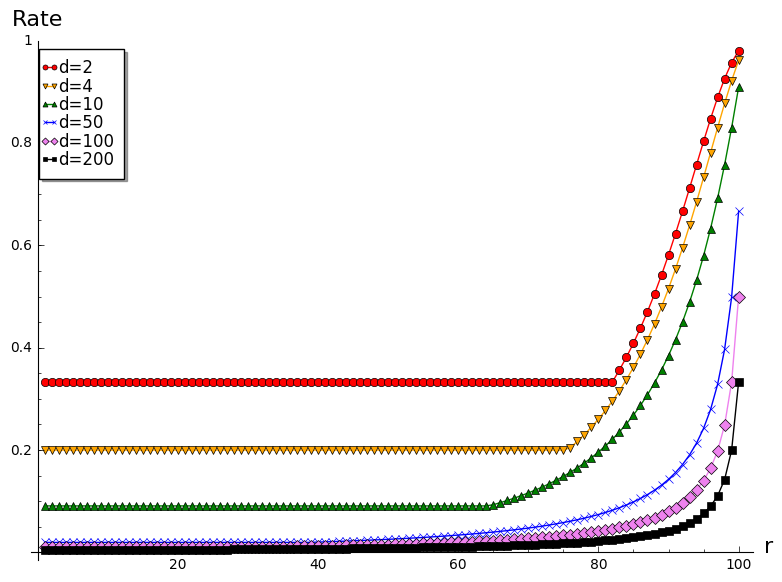}
\caption{The upper bound given by Theorem~\ref{thm:main} on the rate $k/n$ of a $(r,r,d)$-\codename when $k = 100$ and $d = 2,4,10,50,100,200$.  The upper bound is attained at the flat initial segments (by the repetition code) and at $r=k$ (by an MDS code).} 
\label{fig:plt}
\end{center}
\end{figure}

Next, we consider the setting is where $d = r$.  This is the case in the traditional (systematic) LRC set-up, where the $d$ erasures are just the $r$ symbols that need to be repaired.  In this case, \eqref{eq:condition} reads $r \leq (\sqrt{2} - 1)k \approx .414\cdot k$. Thus, in this setting, $r$ must be $\Omega(k)$ before we can
do any better than the repetition code.

Finally we consider the setting, inspired by the traditional LDC set-up, when there are $d = \delta n$ erasures in the codeword.  
In this case, when $r = k$ we can have rate $k/n = 1 - \delta$ using an MDS code, while the repetition code gives vanishing rate $k/n = \Theta(1/n)$.
Theorem~\ref{thm:main} implies that even for $r = (1 - \eps)k$, for any constant $\eps > 0$, we still have $k/n = O(1/n)$.
Indeed, in this setting Theorem~\ref{thm:main} says that
\[ n \geq k(d+1) - d(r-1) = k(d+1) - d((1-\eps)k - 1) \geq \eps k (d+1), \]
and setting $d = \delta n$ the above implies that the rate satisfies $k/n = O(1/n)$.
This reasoning holds even when $\eps = o(1)$.  For example, even if $r = \inparen{1 - \frac{\log(n)}{n}}\cdot k$, Theorem~\ref{thm:main} still implies that $k/n = o(1)$, while if $r = k$ an MDS code has $k/n = 1 - \delta$.

\section{Proof of Theorem~\ref{thm:main}}\label{sec:proof}
Consider a linear $(r,r,d)$-\codename $\cC$ with generator matrix $G$.  
We begin with a lemma which shows that any $(r,r,d)$-\codename must have an extremely structured generator matrix.
\begin{lemma}\label{lem:1} For any $I$ and $J$ as in  Definition~\ref{def:main}, $G|_{I, J}$ is full rank and $G|_{[k]\setminus I, J} = 0$.
\end{lemma}
\begin{proof}
In order for  $x|_{I}$ to be recoverable from $\mathcal{C}(x)|_{J}$ for all $x \in \F^k$, there must exist an injective function $f_{I, J}:\F^r \to \F^r$ such that $f_{I, J}(\mathcal{C}(x)|_{J}) = x|_{I}$ for all $x\in \F^k$. Since $f_{I, J}$ is injective and its domain and codomain are the same size, $f_{I, J}$ is bijective.

It follows that if $x|_I = y|_I$, then $ \mathcal{C}(x)|_{J} =  \mathcal{C}(y)|_{J}$.  Using the definition of $G$, this means that for all $x,y$ with $x|_I = y|_I$, we have
$xG|_{[k], J} = yG|_{[k], J}$ and hence $(x - y)G|_{[k], J} = 0.$
Since $(x - y)|_I = 0$, this implies that for all $x,y \in \F^k$,
$$(x - y)|_{[k]\setminus I}G|_{[k]\setminus I, J} = 0.$$
However, since $(x - y)|_{[k]\setminus I}$ can be any element of $\F^{k - r}$, it must be the case that $G|_{[k]\setminus I, J} = 0.$

Thus, $\cC(x)|_J = (xG)|_J = x|_I G|_{I,J}$.  Since $f_{I,J}$ is invertible and $f_{I, J}(\cC(x)|_J) = f_{I, J}(x|_I G|_{I,J}) = x|_{I}$, $G|_{I,J}$ must be invertible as well.
\end{proof}

Next, we use Lemma~\ref{lem:1} to ``improve" any $(r,r,d)$-\codenamens.  We will repeatedly apply Lemma~\ref{lem:2} until we arrive at a code with good distance and rate, and then apply the Singleton bound to obtain our final bound on $n$.
\begin{lemma}\label{lem:2}
Suppose $\mathcal{C}:\F^k \to \F^n$ is a linear $(r,r,d)$-\codenamens, such that $n < k(d + 1)$.
Then there exists a linear $(r,r,d)$-\codename $\cC':\F^{k-1} \to \F^{n - (d+1) - (k-r)}$.  
\end{lemma}
\begin{proof}
Let $\cC$ be as in the statement of the lemma, with generator matrix $G \in \F^{k \times n}$.
Then there must exist some $i \in [k]$ such that multiples of $e_i$ appear at most a total of $d$ times among the columns of $G$; otherwise, we have $n \geq k(d+1)$, which we assumed was not the case.
Let $T_i \subset [n]$ be the set of indices $j$ so that $G|_{[k],\{j\}} = ge_i$ for some $g \in \F$; so the above reasoning implies that $|T_i| \leq d$.
Let $S_i \subset [n]$ be the set of indices $j$ so that $G_{i,j} \neq 0$.
Notice that $|S_i| \geq d + 1$; otherwise, an adversary could choose a set $D$ of size $d$ so that $S_i \subseteq D$, and then we would never be able to recover $x_i$.
\begin{claim}
$|S_i| \geq d + 1 + k - r.$
\end{claim}
\begin{proof}
Choose $D \subseteq [n]$ so that $|D| = d$ and so that $T_i \subseteq D \subseteq S_i$.
As in Definition~\ref{def:main}, $D$ represents the adversarial set of $d$ erasures.   
We'll show that $|S_i \setminus D| \geq k - r + 1$, which implies the claim. 

Construct a set $X \subseteq [k]\setminus\{i\}$ as follows. For each column $j \in S_i \setminus D$, arbitrarily choose some row $i_j \neq i$ such that $G_{i_j, j} \neq 0$. Let $X = \{i_j \mid j \in S_i \setminus D\}$. Note that we allow for indices $_j$ to be repeated, so we have $|X| \leq |S_i \setminus D|$.
We now claim that 
$$|[k]\setminus X| \leq r - 1.$$
Suppose for the sake of contradiction that $|[k]\setminus X| \geq r$. Then choose some set $I$ such that $I \subseteq [k]\setminus X$, $|I| = r$ and $i \in I$. Since $G$ is the generator matrix of a $(r,r,d)$-\codenamens, there exists $J \subseteq [n]\setminus D$ of size $r$ such that $(xG)|_J$ determines $x|_I$. Since $i \in I$, it follows that there exists $j \in J$ such that $j \in S_i \setminus D$. Furthermore, by Lemma~\ref{lem:1} we have 
$$ G|_{[k] \setminus I, \{j\}} = 0$$
but by construction, $X\subseteq [k] \setminus I$.  This implies that $G|_{X,\{j\}} = 0$ which implies that $G|_{i_j,\{j\}} = 0$.  This is a contradiction of the choice of $i_j$.
Hence, 
$$|[k]\setminus X| \leq r - 1 \implies |X| \geq k - r + 1 \implies |S_i \setminus D| \geq k - r + 1$$
and this proves the claim. 
\end{proof}
Given the claim, we can prove the lemma.  Indeed, 
$G|_{[k]\setminus \{i\}, [n]\setminus S_i}$ is the generator matrix of a $(r,r,d)$-\codename $\cC':\F^{k-1} \to \F^{n - |S_i|}$, and the claim implies that $n - |S_i| \leq n - (d+1) - (k-r)$.  If $n - |S_i| < n - (d+1) - (k-r)$ is strictly smaller, then we can always artificially increase $n$ by padding the codeword with zeros, so this proves the lemma.
\end{proof}

Now by repeatedly applying Lemma~\ref{lem:2}, we will prove Theorem~\ref{thm:main}.
Suppose there exists a $(r,r,d)$-\codename $\cC:\F^k \to \F^n$ with $n < k(d+1)$.

For $\lambda = 0,1,2,\ldots,$ define $n_\lambda$ recursively by $n_0 = n$ and 
\begin{equation}\label{eq:recdef}
n_{\lambda + 1} = n_{\lambda} - (d + 1) - (k  - r - \lambda)
\end{equation} 
for $\lambda \geq 0$. Then
\begin{align*}
n_\lambda &= n - \lambda(d + 1) - \sum_{\ell = 0}^{\lambda - 1} (k - r - \ell) \\
&= n - \lambda(d+1) - \lambda\inparen{ k - r - \frac{\lambda - 1}{2}}.
\end{align*}
Note that by \eqref{eq:recdef}, for $\lambda \leq k - r$, we have that $n_{\lambda} - n_{\lambda + 1} \geq d + 1$. Thus, since we have by assumption that $n = n_0 < k(d+1)$, it holds for all $\lambda$ up to $\lambda = k -r$ that $n_{\lambda} < (k-\lambda)(d+1)$. 

Hence, for any $\lambda \leq k - r$, we can apply Lemma~\ref{lem:2} $\lambda$ times to obtain a code
$\cC'_\lambda$, which is an $(r,r,d)$-\codename with block length $n_\lambda$ and message length $k - \lambda$.

Now consider $\lambda = k - r$. Then $\cC'_\lambda$ is an $(r,r,d)$-\codename with distance $d$ and dimension $k - \lambda = r$.
Then
the Singleton bound applied to $\cC'_\lambda$ implies that $n_{k-r} \geq r + d$, and so 
\begin{align*}
n &\geq (k - r)(d + 1) + \frac{(k - r)(k - r + 1)}{2} + r + d\\
&= k(d + 1) - d(r - 1) + \frac{(k - r)(k - r + 1)}{2}\\
&\geq  k(d + 1) - d(r - 1) + \frac{1}{2}(k - r)^2,
\end{align*}
which is \eqref{eq:bound}.

Since the requirement on $\cC$ was that $n < k(d+1)$, we conclude that for any $(r,r,d)$-\codename $\cC:\F^k \to \F^n$, we have
\[ n \geq \min \{ k(d+1), k(d+1) - d(r-1) + \frac{1}{2} (k-r)^2 \}. \]
Indeed, if $n < k(d+1)$ then we have shown above the \eqref{eq:bound} holds.  And on the other hand if that does not hold then obviously $n \geq k(d+1)$.
The theorem follows.

\section{Conclusion}  
Our goal in this work was to begin a systematic investigation of Question~\ref{q:main}: if we wish to recover $r$ symbols in a message, at what point does simply appealing repeatedly to the $r=1$ case become wasteful?  Clearly one cannot do better for $r=1$, and clearly one can do much better for $r=k$.  When does this transition occur?  At what point can we begin to take advantage of multiple requests? 

In order to investigate this question we introduced \codenamens{s}.
We have shown that, in the special case where $r = m$, $(r,r,d)$-\codenamens{s} that are better than the repetition code do not exist until $r$ is quite large, $\Omega(k)$ in several natural parameter regimes.  This is perhaps surprising, since when $r = k$, one can easily do \em much \em better than repetition using an MDS code.  Thus, in this case, the answer to Question~\ref{q:main} is pessimistic: one cannot handle $r$ requests more efficiently than repeating a single request $r$ times until $r$ is quite large.

The natural question left open by our work is whether or not these same impossibility results hold when $m > r$, and it is our hope that the techniques developed here may help shed light on this question.  When $m = \lambda r$ for an integer $\lambda \geq 1$, one benchmark is given by $n = k\left( d/\lambda + 1 \right),$ which is achieved by breaking the message up into $k/\lambda$ blocks of size $\lambda$, and encoding each with a Reed-Solomon code which can tolerate $d$ errors.  (Notice that when $\lambda = 1$ this is just the repetition code, so for $\lambda \geq 1$ this generalizes the question studied in the current paper).  The open question is thus: for $\lambda > 1$, how large $r$ must be before we can obtain a better rate than $1/(d/\lambda + 1)$?

\IEEEtriggeratref{3} 
\bibliographystyle{plain}
\bibliography{lr}

\end{document}